\documentclass{scrartcl}

\usepackage{lipsum}
\usepackage{amsfonts}
\usepackage{graphicx}
\usepackage{epstopdf}
\usepackage{algorithmic}
\usepackage{cleveref}

\usepackage{amssymb,amsthm}
\usepackage{xspace}
\usepackage{color}
\definecolor{blu3}{rgb}{.1,.0,.4}

\ifpdf
  \DeclareGraphicsExtensions{.eps,.pdf,.png,.jpg}
\else
  \DeclareGraphicsExtensions{.eps}
\fi

\theoremstyle{plain}
\newtheorem{theorem}{Theorem}
\newtheorem{lemma}{Lemma}
\newtheorem{proposition}{Proposition}
\newtheorem{corollary}{Corollary}
\theoremstyle{definition}
\newtheorem{rrule}{Reduction rule}

\crefname{rrule}{Reduction rule}{Reduction rules}
\crefname{lemma}{Lemma}{Lemmas}
\crefname{theorem}{Theorem}{Theorems}
\crefname{corollary}{Corollary}{Corollaries}
\crefname{proposition}{Proposition}{Propositions}
\crefname{section}{Section}{Sections}

\title{The parameterized complexity of finding a 2-sphere in a simplicial complex\thanks{A preliminary version of this paper appeared in Proc. 34th Symposium on Theoretical Aspects of 
Computer Science (STACS 2017)~\cite{BurtonCKP17}.
Partially supported by the Slovenian Research Agency, program P1--0297 and project L7--5459.}}

\author{Benjamin Burton\thanks{School of Mathematics and Physics, The University of Queensland, Brisbane, Australia
  (\texttt{bab@maths.uq.edu.au}).}
\and Sergio Cabello\thanks{Department of Mathematics, IMFM, and Department of Mathematics, FMF, University of Ljubljana, Slovenia 
  (\texttt{sergio.cabello@fmf.uni-lj.si}).}
\and Stefan Kratsch\thanks{Department of Computer Science, University of Bonn, Bonn, Germany. Current address: Department of Computer Science, Humboldt-Universit\"at zu Berlin, Berlin, Germany (\texttt{kratsch@informatik.hu-berlin.de}).}
\and William Pettersson\thanks{School of Science, RMIT University, Melbourne, Australia. Current address: School of Computing Science, University of Glasgow, Glasgow, United Kingdom
  (\texttt{william@ewpettersson.se}).}
}

\usepackage{amsopn}

\newcommand{\C}{\ensuremath{\mathcal{C}}\xspace} % caligraphic C
\newcommand{\D}{\ensuremath{\mathcal{D}}\xspace} % caligraphic D
\newcommand{\K}{\ensuremath{\mathcal{K}}\xspace} % simplicial complex
\newcommand{\Oh}{\ensuremath{\mathcal{O}}\xspace} % caligraphic O
\newcommand{\Q}{\ensuremath{\mathcal{Q}}\xspace} % caligraphic Q
\renewcommand{\S}{\ensuremath{\mathcal{S}}\xspace} % caligraphic S
\newcommand{\T}{\ensuremath{\mathcal{T}}\xspace} % caligraphic T
\newcommand{\RR}{\ensuremath{\mathbb R}\xspace}  % real numbers
\newcommand{\ZZ}{\ensuremath{\mathbb Z}\xspace}  % integer numbers
\newcommand{\NN}{\ensuremath{\mathbb N}\xspace}  % natural numbers
\renewcommand{\SS}{\ensuremath{\mathbb S}\xspace}  % sphere
\renewcommand{\P}{\ensuremath{\mathsf{P}}\xspace}
\newcommand{\NP}{\ensuremath{\mathsf{NP}}\xspace}
\newcommand{\FPT}{\ensuremath{\mathsf{FPT}}\xspace}
\newcommand{\W}[1]{\ensuremath{\mathsf{W[#1]}}\xspace}
\newcommand{\Sd}{\ensuremath{\mathsf{Sd}}\xspace}
\DeclareMathOperator{\tsquare}{top}
\DeclareMathOperator{\lsquare}{left}
\DeclareMathOperator{\bsquare}{bottom}
\DeclareMathOperator{\rsquare}{right}
\DeclareMathOperator{\csquare}{center}

\newcommand{\twosphere}{\textsc{2-dim-sphere}\xspace}
\newcommand{\deletiontotwosphere}{\textsc{Deletion-to-2-dim-sphere}\xspace}
\newcommand{\gridtiling}{\textsc{Grid Tiling}\xspace}

\def\DEF#1{\emph{#1}}

\begin{document}

\maketitle

\begin{abstract}
  We consider the problem of finding a subcomplex $\K'$ of a simplicial complex $\K$ such that $\K'$ is homeomorphic to the 2-dimensional sphere, $\SS^2$. We study two variants of this problem. The first asks if there exists such a $\K'$ with at most $k$ triangles, and we show that this variant is \W{1}-hard and, assuming ETH, admits no $\Oh(n^{o(\sqrt{k})})$ time algorithm. We also give an algorithm that is tight with regards to this lower bound. The second problem is the dual of the first, and asks if $\K'$ can be found by removing at most $k$ triangles from $\K$. This variant has an immediate $\Oh(3^{k}poly(|\K|))$ time algorithm, and we show that it admits a polynomial kernelization to $\Oh(k^2)$ triangles, as well as a polynomial compression to a weighted version with bit-size $\Oh(k \log k)$.
\end{abstract}

\section{Introduction}

Topology is the study of the properties of spaces that are preserved under continuous deformations of the space.
Intuitively, this can be summed up by the joke description of a topologist as a mathematician who
cannot tell the difference between a coffee mug and a doughnut, as each can be continuously deformed into the other.
In this paper we discuss manifolds, which are topological spaces that locally look like Euclidean space.
That is to say, every point in a $d$-manifold (without boundary) has a neighborhood homeomorphic to $\RR^d$.

The simplest manifold is the {\em $d$-sphere}, which is the boundary of a $(d+1)$-dimensional ball, where
the $(d+1)$-dimensional ball is simply a closed neighborhood of $\RR^{d+1}$.
In particular, the $2$-sphere which we will discuss is the $2$-dimensional surface of a $3$-dimensional ball (such as a soccer ball) that would live in the 3-space of our physical world.
The sphere is of interest as it relates to the {\em connected sum\/} operation on manifolds.
A connected sum of two $d$-manifolds is found by removing a $(d+1)$-dimensional ball from each manifold, and identifying the two components along the boundaries of the respective balls.
The $d$-sphere forms the identity element of this operation.
Finding embedded $d$-spheres that can separate a manifold into two non-trivial components is therefore the topological equivalent of the factorization of integers.
Indeed, a manifold that has no such spheres is called {\em prime}, and a {\em prime decomposition\/} of a manifold is a decomposition into prime manifolds. 

In this paper we will use {\em (abstract) simplicial complexes\/} to combinatorially represent manifolds.
At an informal level, a simplicial complex is a collection of simplices that are glued by identifying some faces.
In principle, the abstract simplicial complex does not live in any ambient space, although we can always represent it geometrically using spaces of high enough dimension.
A formal definition is given in \cref{section:preliminaries}.

Arguably, the most natural question to ask regarding a simplicial complex is whether it represents a manifold. 
The question is easy to answer for 2-manifolds: it suffices to check whether each edge is adjacent to exactly two triangles.
Additionally, in two dimensions we can recognize the manifold by calculating the Euler characteristic of the simplicial complex, itself a simple enumeration of vertices, edges and faces, and checking whether it is orientable.
Recognizing the manifold of a simplicial 3-complex is far harder, even for the 3-sphere~\cite{Schleimer11}. The recognition of 4-dimensional manifolds and the 5-sphere is an undecidable problem (see for example the appendix of~\cite{Nabutovsky1995}), while the recognition of the 4-sphere is a notorious open problem.
Interestingly, in dimensions 4 and higher there exists manifolds (such as the $E_8$ manifold) which can not even be represented as a simplicial complex~\cite{Freedman82}.

\subparagraph*{Our work.}
We return to a basic problem for 2-dimensional simplicial complexes: does a given simplicial complex contain a subcomplex that is (homeomorphic to) a 2-sphere? The problem is known to be \NP-hard, 
and we study its parameterized complexity with respect to the solution size (number of triangles in the subcomplex) and its dual (number of triangles not in the subcomplex); we begin with the former problem.

\begin{quote}
	\twosphere\\
    \textbf{Input:} A pair $(\K,k)$ where 
		$\K$ is a 2-dimensional simplicial complex and 
    	$k$ is a positive integer.\\
    \textbf{Question:} Does $\K$ contain a subcomplex with at most $k$ triangles
    	that is homeomorphic to the 2-dimensional sphere?
\end{quote}

We show that this problem is \W{1}-hard with respect to $k$. In fact we show that, 
assuming the Exponential Time Hypothesis (ETH;\@ see preliminaries), the problem cannot be solved in
$n^{o(\sqrt{k})}$ time. ETH implies a core hypothesis of parameterized complexity, namely that $\FPT\neq\W{1}$ (comparable to the hypothesis that $\P\neq\NP$). Together with its twin SETH (the Strong Exponential Time Hypothesis) it is known to imply a wide range of lower bounds, often matching known algorithmic results, for various \NP-hard problems. (To note, a very active branch of research uses SETH for tight lower bounds for problems in $\P$.)

\begin{theorem}\label{theorem:main:lowerbound}
The \twosphere problem is \W{1}-hard with respect to parameter $k$ and, unless ETH fails, it has no $f(k)n^{o(\sqrt{k})}$-time algorithm for any computable function $f$.
\end{theorem}

Note that the related problem variant of finding a subcomplex with \emph{at least $k$} triangles that is homeomorphic to the 2-dimensional sphere is \NP-hard for $k=0$, as this is simply \NP-hard problem of testing whether there is any subcomplex that is homeomorphic to the $2$-sphere. (Note that hardness for finding a subcomplex with at most $k$ triangles also implies hardness for finding one with \emph{exactly $k$} triangles.)

We complement \cref{theorem:main:lowerbound} by giving an algorithm for \twosphere that runs in $n^{\Oh(\sqrt{k})}$ time, which is essentially tight; it can also be used to find a solution with exactly $k$ triangles.

\begin{theorem}\label{theorem:main:upperbound}
The \twosphere problem can be solved in time $2^{\Oh(k)}n^{\Oh(\sqrt{k})}$.
\end{theorem}

For the dual problem, we are interested in the parameterized complexity relative to the number $k$ of triangles that are not in the solution (i.e., not in the returned subcomplex that is homeomorphic to the $2$-sphere). In other words, the question becomes that of deleting $k$ triangles (plus edges and vertices that are only incident with these triangles) to obtain a subcomplex that is homeomorphic with the $2$-sphere. Similarly to before, deleting \emph{at least $k$} triangles is \NP-hard for $k=0$ as that is just asking for existence of any subcomplex that is homeomorphic to the $2$-sphere. We consider the question of deleting \emph{at most $k$} triangles.

\begin{quote}
	\deletiontotwosphere\\
    \textbf{Input:} A pair $(\K,k)$ where 
		$\K$ is a 2-dimensional simplicial complex and 
    	$k$ is a positive integer.\\
    \textbf{Question:} Can we delete at most $k$ triangles in $\K$ so that the remaining
    	subcomplex is homeomorphic to the 2-dimensional sphere?
\end{quote}

There a simple $\Oh(3^{k}poly(|\K|))$ time algorithm for this problem: While there is an edge incident with at least three triangles, among any three of these triangles at least one must be deleted. Recursive branching on these configurations gives rise to search tree with at most $3^k$ leaves, each of which is an instance with (1) $k=0$ and at least one edge is shared by at least three triangles, or (2) $k\geq 0$ and each edge is shared by at most two triangles. The former instances can clearly be discarded, the latter can be easily solved in polynomial time: components (with enough connectivity) and a boundary can be discarded (updating budget $k$ accordingly); components without boundary have each edge being shared by exactly two triangles and we can efficiently test which ones are homeomorphic to the $2$-sphere (keeping the largest).

Knowing, thus, that \deletiontotwosphere is \emph{fixed-parameter tractable} for parameter $k$, we ask whether it admits a polynomial kernelization or compression, i.e., an efficient preprocessing algorithm that returns an equivalent instance of size polynomial in $k$. We prove that this is the case by giving, in particular, a compression to almost linear bit-size.

\begin{theorem}\label{theorem:other:compression}
The \deletiontotwosphere problem admits a polynomial kernelization to instances with $\Oh(k^2)$ triangles and bit-size $\Oh(k^2\log k)$ and a polynomial compression to weighted instances with $\Oh(k)$ triangles and bit-size $\Oh(k\log k)$.
\end{theorem}

\subparagraph*{Related work.}
A sketch of NP-hardness for the \twosphere problem was 
given by Ivanov~\cite{MO} in a Mathoverflow question.

Our work is one of the few ones combining topology and fixed parameter tractability.
In this direction there have been recent results focused on algorithms in 3-manifold topology~\cite{Bagchi16,BurtonCourcelles,Burton15,Burton14,Maria16}.
The problem of finding a shortest 1-dimensional cycle $\ZZ_2$-homologous to a given cycle
in a 2-dimensional cycle was shown to be NP-hard by Chao and Freedman~\cite{ChenF11}. Erickson and Nayyeri~\cite{EricksonN11a} showed that the problem is fixed-parameter tractable for surfaces, when parameterized by genus of the surface. The result has been extended~\cite{BusaryevCCDW12} to arbitrary 2-dimensional simplicial complexes parameterized by the first Betti number. Finally, let us mention that deciding whether a graph (1-dimensional simplicial complex) can be embedded in surface of genus $g$ is fixed-parameter tractable with respect to the genus~\cite{KawarabayashiMR08,Mohar99}.

\subparagraph*{Organization.}
We begin with preliminaries on computational topology and parameterized complexity (\cref{section:preliminaries}). The proofs for \cref{theorem:main:lowerbound} and~\cref{theorem:main:upperbound} about \twosphere are given in \cref{section:main:hardness} and~\cref{section:main:algorithm}. The preprocessing result for \deletiontotwosphere, i.e., \cref{theorem:other:compression}, is proved in \cref{section:preprocessing}. We conclude in \cref{section:conclusion} with some open problems.

%%%%%%%%%%%%%%%%%%%%%%%%%%%%%%%%%%%%%%%%%%%%%%%%%%%%%%%%%%%%%%%%%%%%%%%%%%%%%%%%%%%%%%%%%%%%%%%%%%%%%%%%%%%%%%%%%%%%%%
\section{Background and notation}\label{section:preliminaries}

For each positive integer $n$ we use $[n]$ to describe the set $\{1,\dots,n\}$.

\subparagraph{Topological background.}
We give a very succinct summary of the topological background we need
and refer the reader to~\cite[Chapter 1]{Matousek07} or~\cite[Chapter 1]{Munkres93a}
for a comprehensive introduction. The results we mention are standard and available in several books.

A {\em homeomorphism\/} between two topological spaces is a continuous mapping between the two spaces whose inverse is also continuous. If such a homeomorphism exists, we say that the two spaces are {\em homeomorphic}. Any topological property is invariant under homeomorphisms.

A {\em $d$-manifold\/} is a topological space where each point has a neighborhood homeomorphic to $\RR^d$ or the closed half-space $\{ (x_1,\dots,x_d)\in \RR^d \mid x_1\ge 0 \}$. A point of the manifold where no neighborhood is homeomorphic to $\RR^d$ is a {\em boundary point}. 
In this paper we focus on 2-manifolds, often called {\em surfaces}, which are locally equivalent to the Euclidean plane or a half-plane.  It is known that the boundary of a (compact) $2$-manifold is the union of finitely many $1$-manifolds (circles). A surface can be described by a collection of triangles and a collection of pairs of edges of triangles that are identified. If each edge appears in some pairing, then the surface has no boundary.

A {\em geometric $d$-simplex\/} is the convex hull of $d+1$ points in $\RR^{d'}$
that are not contained in any hyperplane of dimension $d-1$; this requires $d'\ge d$.
A {\em face\/} of simplex $\sigma$ is a simplex of a subset of the points defining $\sigma$.  
A {\em geometric simplicial complex\/} $\K$ is a collection of geometric simplices where each face of each simplex of $\K$ is also in $\K$, and any non-empty intersection of any two simplices of $\K$ is also in $\K$. The \DEF{carrier} of $\K$, denoted by $|| \K ||$, is the union of all the simplices in $\K$. A geometric simplicial complex $\K$ is a {\em triangulation\/} of $X$ if $X$ and $||\K||$ are homeomorphic.
Quite often we talk about properties of $\K$ when we mean properties of its carrier $||\K||$. For example, we may say that a geometric simplicial simplex $\K$ is homeomorphic to a topological space $X$ when we mean that $||\K||$ and $X$ are homeomorphic.

An {\em (abstract) simplicial complex\/} $\K$ is a finite family of sets with the property that any subset of any set of $\K$ is also in contained $\K$. 
An example of abstract simplicial complex is 
$\{ \emptyset, \{1 \}, \{2 \}, \{3 \}, \{4 \},
\{1,2 \},\{1,3\}, \{1,4 \},\{2,3 \}, \{3,4 \}, 
\{1,3,4\} \}$.
The singletons of $\K$ are called {\em vertices\/} and the set of vertices is denoted by $V(\K)$. We can assume without loss of generality that $V(\K)=[n]$ for some natural number $n$, as we already had in the previous example.
The {\em dimension\/} of the (abstract) simplicial complex $\K$ is 
$\max_{\sigma\in \K} |\sigma|-1$.

In this paper we focus on (abstract) simplicial complexes and \emph{we will remove the adjective ``abstract''} when referring to them.  
Here we are interested in $2$-dimensional simplicial complexes. We can describe them by giving either the list of all simplices or a list of the inclusion-wise maximal simplices. Since in dimension $2$ the length of these two lists differ by a constant factor, the choice is asymptotically irrelevant.
(For unbounded dimensions, this difference is sometimes relevant.)

A {\em geometric realization\/} of a simplicial complex $\K$ is an injection $f\colon V(\K)\rightarrow \RR^{d'}$ such that  
$\{ CH(f(\sigma))\mid \sigma\in \K\setminus\{\emptyset\}\}$ is a geometric simplicial complex, where $CH(\cdot)$ denotes the convex hull. It is easy to show that the carriers of any two geometric realizations of a simplicial complex are homeomorphic.
Abusing terminology, we will talk about properties of a simplicial complex when (the carrier of) its geometric realizations have the property. For example, we say that a simplicial complex $\K$ is triangulation of the 2-sphere when we mean that some geometric realization of $\K$ is a triangulation of the $2$-sphere (and thus all geometric realizations of $\K$ are triangulations of the $2$-sphere).

\subparagraph{Parameterized complexity.} Again, we just provide a very succinct summary. See the books by Cygan et al.~\cite{CyganFKLMPPS15} or by Downey and Fellows~\cite{DowneyF13} for recent comprehensive accounts.

A \emph{parameterized problem} is a language $\Q\subseteq\Sigma^*\times\NN$ where $\Sigma$ is any finite alphabet and $\NN$ denotes the non-negative integers; the second component $k$ of an instance $(x,k)\in\Sigma^*\times\NN$ is called its \emph{parameter}. A parameterized problem $\Q$ is \emph{fixed-parameter tractable} if there is an algorithm $A$, a constant $c$, and a computable function $f\colon\NN\to\NN$ such that $A$ correctly decides $(x,k)\in\Q$ in time $f(k)\cdot|x|^c$ for all $(x,k)\in\Sigma^*\times\NN$.
A \emph{kernelization} of a parameterized problem $\Q$ with \emph{size $h\colon\NN\to\NN$} is a polynomial-time algorithm $K$ that on input $(x,k)\in\Sigma^*\times\NN$ takes time polynomial in $|x|+k$ and returns an instance $(x',k')$ of size at most $h(k)$ such that $(x,k)\in\Q$ if and only if $(x',k')\in\Q$. If $h(k)$ is polynomially bounded then $K$ is a \emph{polynomial kernelization}. If the output of $K$ is instead an instance of any (unparameterized) problem $L'$ then we called it a \emph{(polynomial) compression}.

The prevalent method of showing that a parameterized problem $\Q'\subseteq\Sigma'^*\times\NN$ is \emph{not} fixed-parameter tractable is to give a \emph{parameterized reduction} from a problem $\Q\subseteq\Sigma^*\times\NN$ that is hard for a class called \W{1}, which contains the class $\FPT$ of all fixed-parameter tractable problems; it is assumed that $\FPT\neq\W{1}$. A parameterized reduction from $\Q$ to $\Q'$ is an algorithm $R$ that on input $(x,k)\in\Sigma^*\times\NN$ takes time $f(k)\cdot|x|^c$ and returns an instance $(x',k')\in\Sigma'^*\times\NN$ such that: $(x,k)\in\Q$ if and only if $(x',k')\in\Q'$ and such that $k'\leq g(k)$; here $f,g\colon\NN\to\NN$ are computable functions and $c$ is a constant, all independent of $(x,k)$.
Parameterized reductions can also be used to transfer lower bounds on the running time. A common starting point for this is the \emph{Exponential Time Hypothesis (ETH)} which posits that there is a constant $\delta_3>0$ such that no algorithm solves \textsc{$3$-SAT} in time $\Oh(2^{\delta_3n})$ where $n$ denotes the number of variables. In particular, this rules out subexponential-time algorithms for \textsc{$3$-SAT} and, by appropriate reductions, for a host of other problems.

%%%%%%%%%%%%%%%%%%%%%%%%%%%%%%%%%%%%%%%%%%%%%%%%%%%%%%%%%%%%%%%%%%%%%%%%%%%%%%%%%%%%%%%%%%%%%%%%%%%%%%%%%%%%%%%%%%%%%%
\section{Hardness of 2-dim sphere}\label{section:main:hardness}

In this section we provide a proof for \cref{theorem:main:lowerbound}, namely that \twosphere is \W{1}-hard for parameter $k$ and, under ETH, admits no $\Oh(n^{o(\sqrt{k})})$ time algorithm. To obtain the result we give a polynomial-time reduction from the \gridtiling problem introduced by Marx~\cite{Marx12}.

\begin{quote}
	\gridtiling\\
    \textbf{Input:} A triple $(n,k,\S)$ where $n$ is a positive integer,
    	$k$ is a positive integer, and $\S$ is a tuple of $k^2$ 
        nonempty sets $S_{i,j}\subseteq [n] \times [n]$, 
		where $i,j \in [k]$. \\
    \textbf{Question:} Can we choose for each $i,j \in [k]$ 
    	a pair $(a_{i,j},b_{i,j}) \in S_{i,j}$ such that 
        $a_{i,j}=a_{i,j+1}$ for all $i\in [k]$, $j\in [k-1]$, and
        $b_{i,j}=b_{i+1,j}$ for all $i\in [k-1]$, $j\in [k]$?
\end{quote}

It is convenient to visualize the input elements as displayed in a 
$(k\times k)$-tiled square. The squares are indexed like matrices: the top left tile
corresponds to the index $(i,j)=(1,1)$ and the bottom left tile 
corresponds to the index $(i,j)=(k,1)$. Inside the $(i,j)$-tile we put
the elements of $S_{i,j}$. 
An example of an instance for \gridtiling is given in \cref{fig:gridtiling}.
The task is to select a 2-tuple in each tile such
that the selected elements in each row have the same first coordinate
and the selected elements in each column have the same second coordinate.
The following lower bound is known for \gridtiling.

\begin{figure}
\centering
	\includegraphics[page=1]{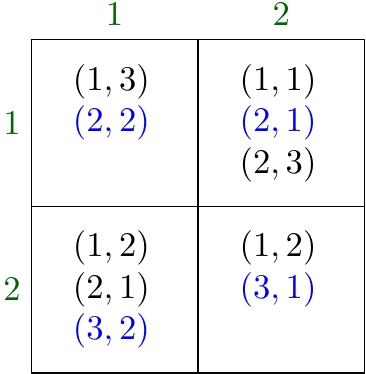}
	\caption{A yes-instance of \gridtiling with $n=3$, $k=2$, 
		and $S_{1,1}=\{ (1,3), (2,2)\}$; the blue entries constitute a solution.}
	\label{fig:gridtiling}
\end{figure}

\begin{theorem}[\cite{Marx12}]\label{theorem:marx:gridtilinglowerbound}
\gridtiling is W[1]-hard and, unless ETH fails, it has no $f(k)n^{o(k)}$-time algorithm for any computable function $f$.
\end{theorem}

Consider an instance $(n,k,\S)$ of \gridtiling. We are going to construct
an equivalent instance $(\K,k')$ to \twosphere where $k'=\Theta(k^2)$. 

Let $\sigma$ be the simplicial complex shown in \cref{fig:square}, left.
It is a triangulation of a square with a middle vertex, denoted $\csquare(\sigma)$.
We denote the consecutive 2-edge paths on the boundary as $\lsquare(\sigma)$, $\tsquare(\sigma)$, $\rsquare(\sigma)$ and $\bsquare(\sigma)$.
The orientation of the path, indicated with an arrow,
defines the way we glue in later steps of the construction.
In our figures we will always orient the squares to match these names in the intuitive way.
For our construction, the important property of $\sigma$ is that there is no triangle containing $\csquare(\sigma)$ and one boundary edge and that there 
is no triangle containing two boundary edges of $\sigma$.

\begin{figure}
\centering
	\includegraphics[page=2,scale=.9]{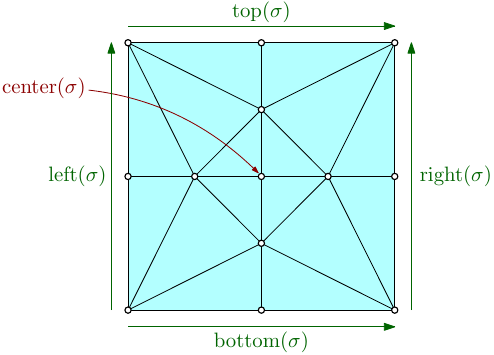}\label{fig:square}
\hfill
	\includegraphics[page=3,scale=.9]{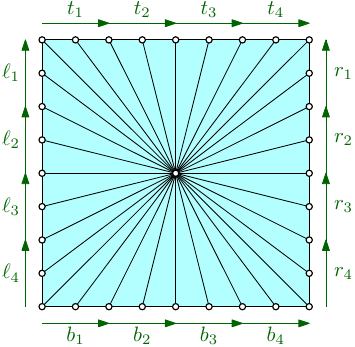}\label{fig:back}
	\caption{Left: The triangulated square $\sigma$.
			Right: back sheet when $k=4$.}
\end{figure}

For each $(a,b)$ in each $S_{i,j}$ we make a new copy of $\sigma$ and
denote it by $\sigma(a,b,i,j)$. We make some identifications, 
according to the following rules; \cref{fig:identifications} may be useful:
\begin{itemize}
    \item For each $i\in [k]$, $j\in [k-1]$, and $a\in [n]$,
		we identify together all the 2-edge paths 
		$\rsquare(\sigma(a,b,i,j))$, where $(a,b)\in S_{i,j}$,
		and all the 2-edge paths $\lsquare(\sigma(a,b',i,j+1))$, where $(a,b')\in S_{i,j+1}$.
		Thus, for each $i,j,a$ we have identified 
		$|\{ b\in [n]\mid (a,b)\in S_{i,j}\}|+|\{ b'\in [n]\mid (a,b')\in S_{i,j+1}\}|$
		2-edge paths into a single one.
    \item For each $i\in [k-1]$, $j\in [k]$, and $b\in [n]$,
		we identify together all the 2-edge paths  
		$\bsquare(\sigma(a,b,i,j))$, where $(a,b)\in S_{i,j}$,
		and all the 2-edge paths  $\tsquare(\sigma(a',b,i+1,j))$, 
		where $(a',b)\in S_{i+1,j}$.
		Thus, for each $i,j,b$ we have identified 
		$|\{ a\in [n]\mid (a,b)\in S_{i,j}\}|+|\{ a'\in [n]\mid (a',b)\in S_{i+1,j}\}|$
		2-edge paths into a single one.
	\item For each $i,j\in [n]$, we identify the vertices $\csquare(\sigma(a,b,i,j))$ 
		over all $(a,b)\in S_{i,j}$. Thus, we identified $|S_{i,j}|$ vertices into a 
		single one.
\end{itemize}

\begin{figure}[h]
\centering
	\includegraphics[page=4,scale=.7]{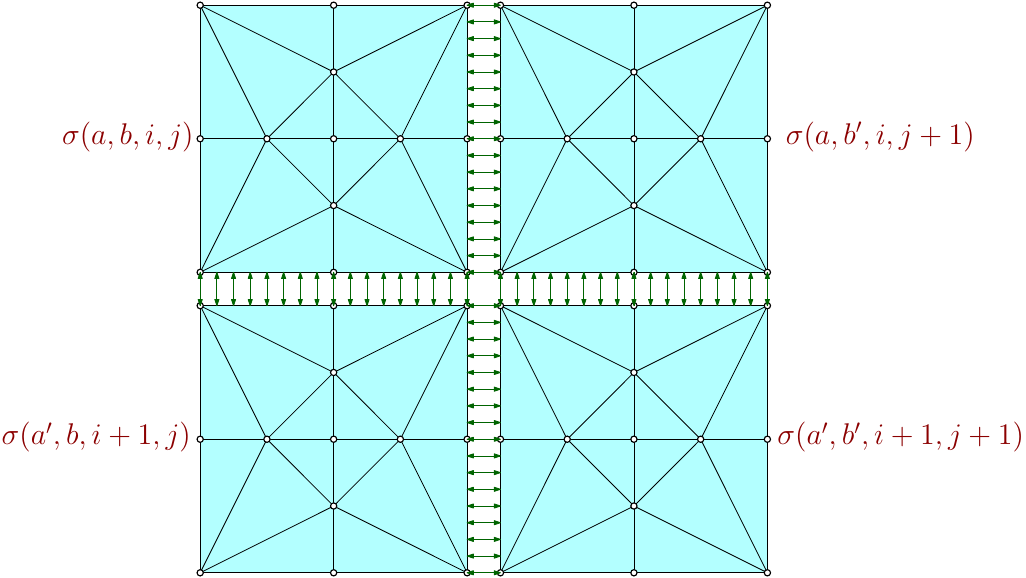}
	\caption{Example showing some identifications of 2-edge paths.}
	\label{fig:identifications}
\end{figure}

To finalize the construction, we triangulate a square such that it has $2k$ edges on
each side, as shown in \cref{fig:back}, right. 
We will refer to this simplicial complex as the \emph{back sheet}.
We split the boundary of the square into $2$-edge paths and label them,
in a clockwise traversal of the boundary of the square, by
$t_1,\dots,t_k$, $r_1,\dots,r_k$, $b_k,\dots,b_1$, and $\ell_k,\dots,\ell_1$. 
(We use $t$ as intuition for top, $r$ as intuition for right, etc.) 
Note that the indices for $b$ and $\ell$ run backwards. 
In the figure we also indicate the orientation of the 2-edge paths,
that are relevant for the forthcoming identifications.

Then we make the following additional identifications.
\begin{itemize}
    \item For each $i\in [n]$ and each $(a,b)\in S_{i,1}$, 
		we identify $\lsquare(\sigma(a,b,i,1))$ and $\ell_i$. 
    \item For each $i\in [n]$ and each $(a,b)\in S_{i,k}$, 
		we identify $\rsquare(\sigma(a,b,i,k))$ and $r_i$. 
    \item For each $j\in [n]$ and each $(a,b)\in S_{1,j}$, 
		we identify $\tsquare(\sigma(a,b,1,j))$ and $t_j$. 
    \item For each $j\in [n]$ and each $(a,b)\in S_{k,j}$, 
		we identify $\bsquare(\sigma(a,b,k,j))$ and $b_j$. 
\end{itemize}

Note that whenever we identify the endpoints of two edges in pairs,
we also identified the edges. Thus, we have constructed 
a simplicial complex. (If it is not obvious to the reader 
that we have a simplicial complex, we could always use barycentric
subdivisions, which will be introduced below,
to ensure that indeed we have a simplicial complex.)
Let $\K=\K(n,k,\S)$ denote the resulting simplicial complex.
Set $k'=16\cdot k^2 + 8k$. With the following lemmas we prove that $(\K,k')$ is yes for \twosphere if and only if $(n,k,\S)$ is yes for \gridtiling.

\begin{lemma}
\label{lemma:gridtiling-to-2sphere}
If $(n,k,\S)$ is a yes-instance for \gridtiling, then $\K$ contains a subcomplex with $k'$ triangles that is homeomorphic to the $2$-sphere.
\end{lemma}

\begin{proof}
	Because $(n,k,\S)$ is a yes-instance for \gridtiling,
	there exist pairs $(a_{i,j},b_{i,j})\in S_{i,j}$, where $i,j\in [k]$,
	such that $a_{i,j}=a_{i,j+1}$ (for $i\in [k],j\in [k-1]$)
	and $b_{i,j}=b_{i+1,j}$ (for $i\in [k-1],j\in [k]$).
	
	Consider the subcomplex $\widetilde\K$ of $\K$ induced by the squares
	$\sigma(a_{i,j},b_{i,j},i,j)$, where $i,j\in [k]$.
	During the identifications we have glued the square $\sigma(a_{i,j},b_{i,j},i,j)$
	to the square $\sigma(a_{i,j+1},b_{i,j+1},i,j+1)$ when making the identification
	$\rsquare(\sigma(a_{i,j},b_{i,j},i,j))= \lsquare(\sigma(a_{i,j+1},b_{i,j+1},i,j+1))$
	because $a_{i,j}=a_{i,j+1}$ (for $i\in [k],j\in [k-1]$).
	Similarly, we have glued $\sigma(a_{i,j},b_{i,j},i,j)$
	to $\sigma(a_{i+1,j},b_{i+1,j},i+1,j)$ when making the identification
	$\bsquare(\sigma(a_{i,j},b_{i,j},i,j))= \tsquare(\sigma(a_{i+1,j},b_{i+1,j},i+1,j))$
	because $b_{i,j}=b_{i+1,j}$ (for $i\in [k-1],j\in [k]$).
	Thus $\widetilde\K$ is a ``big square'' obtained by gluing $k^2$ copies of $\sigma$
	in a $(k\times k)$-grid-like way.
	Together with the back sheet, that is glued to the boundary of $\widetilde\K$,
	we get a triangulation of the $2$-sphere. 
	Since each square $\sigma(\cdot)$ has $16$ triangles and the back sheet has
	$8k$ triangles, the resulting triangulation has $16k^2 + 8k$ triangles.

	While we already claimed it and it is intuitively clear that the manifold we constructed is a $2$-sphere, 
	a formal argument can be carried out showing that this triangulation has Euler characteristic 2.
	For this we have to count the number of vertices and edges of the triangulation, which we do
	as we ``build'' the manifold adding the $k^2$ squares on the front and then adding the back sheet. 
	For the first of these ($\sigma(a_{1,1},b_{1,1},1,1)$), we can count 13 vertices and 28 edges.
	For each $1 < j \leq k$, $\sigma(a_{1,j},b_{1,j},1,j)$ (and respectively $\sigma(a_{j,1},b_{j,1},j,1)$) has 10 vertices and 
	26 edges as the left-most (respectively top-most) vertices and edges are counted in an earlier square.
	Similarly, the remaining $(k-1)^2$ squares $\sigma(a_{i,j},b_{i,j},i,j)$, for $i,j > 1$, have 8 vertices and 24 edges.
	Lastly, the rear square, the back sheet, has $1$ vertex and $8k$ edges as the ``outer'' edges are already 
	counted when we considered the front squares.
	This gives a total of 
	\[ 
		1 + 13 + 2(k-1)\times 10 + (k-1)^{2}\times 8 ~=~ 8k^2 + 4k + 2
	\] vertices and 
	\[
		8k + 28 + 2(k-1)\times 26 + (k-1)^2\times 24 ~=~ 24k^2 + 12k
	\] edges.
	Thus, the Euler characteristic is
	\[
		(8k^2 + 4k + 2) - (24k^2 + 12k) + (16k^2 + 8k) ~=~ 2. \qed
	\]
\end{proof}

\begin{lemma}\label{lemma:2sphere-to-gridtiling}
  If $\K$ contains a subcomplex $\K'$ homeomorphic to the $2$-sphere,
  then $(n,k,\S)$ is a yes-instance for \gridtiling.
\end{lemma}

\begin{proof}
  We show this by first demonstrating that, for any pair $i,j\in[k]$, the subcomplex $\K'$ cannot contain two distinct squares $\sigma(a,b,i,j)$ and $\sigma(c,d,i,j)$.
  We then show that for any pair $i,j\in[k]$, at least one of the squares $\sigma(a,b,i,j)$ must be part of $\K'$.
  Lastly we combine these two facts to construct a solution for the \gridtiling instance $(n,k,\S)$.

  We begin by noting that $\K'$ cannot be empty. Moreover, note that, if for any
  values of $a,b,i,j$ the subcomplex $\K'$ contains one triangle from $\sigma(a,b,i,j)$, then
  $\K'$ must contain all triangles from $\sigma(a,b,i,j)$. This follows from
  the fact that the 2-dimensional sphere has no boundary and the interior edges of $\sigma(a,b,i,j)$
  are not shared by any other triangles. In the rest of this
  argument, we need only consider whether $\K'$ does or does not contain all of
  $\sigma(a,b,i,j)$ for any values of $a,b,i,j$.

  Now assume that, for some pair $i,j\in [k]$,  the subcomplex $\K'$ contains two distinct squares $\sigma(a,b,i,j)$ and $\sigma(c,d,i,j)$ 
  and consider the neighborhood of the point $\csquare(\sigma(a,b,i,j))$.
  Since $\csquare(\sigma(a,b,i,j))$ was identified with $\csquare(\sigma(c,d,i,j))$, 
  we see that this point has no neighborhood homeomorphic to a plane, and so clearly $\K'$ cannot contain both of these distinct squares.

  We now show that, for any pair $i,j\in[k]$, $\K'$ must contain some square of the form $\sigma(a,b,i,j)$.
  If $\K'$ contains no squares $\sigma(\cdot)$ at all then it can only contain either
  part of, or the whole of, the back sheet but either way $\K'$ cannot be a
  2-sphere. Thus we know that $\K'$ must contain $\sigma(a,b,i,j)$
  for at least one set of values $a,b,i,j$. Given this, assume we have a pair
  $i',j'\in[n]$ such that $\sigma(a',b',i',j')$ is not in
  $\K'$ for each pair $(a',b')\in S_{i',j'}$. Let $\sigma(a,b,i,j)$ be a square in $\K'$ for
  some pair $a,b$, and without loss of generality, assume that $i' = i - 1$ and
  $j' = j$.  This is equivalent to choosing two adjacent cells where one
  contains a square in $\K'$ and the other does not contain any square in
  $\K'$, and can always be achieved by appropriate selection of values (and
  possibly rotating or flipping the whole construction).

  Consider an edge of the 2-edge path $\tsquare(\sigma(a,b,i,j))$ in $\K'$. Since $\K'$ is a
  2-dimensional sphere, this edge cannot be a boundary and thus must separate
  two distinct triangles. One of these triangles is present in $\sigma(a,b,i,j)$.
  By our construction, the other triangle is either in $\sigma(a',b,i-1,j)$
  (where $\tsquare(\sigma(a,b,i,j))$ is identified with
  $\bsquare(\sigma(a',b,i-1,j))$) or in $\sigma(c,b,i,j)$ with $a \neq c$
  (where $\tsquare(\sigma(a,b,i,j))$ is identified with
  $\tsquare(\sigma(c,b,i,j))$).  This means that one of $\sigma(c,b,i,j)$ or
  $\sigma(a',b,i-1,j)$ must be in $\K'$.
  By our earlier argument, $\sigma(a,b,i,j)$ and $\sigma(c,b,i,j)$ cannot both be in $\K'$.
  This means that $\sigma(a',b,i-1,j)$ must be in $\K'$, and thus our assumption
  must be false. Therefore for each pair $i,j\in[n]$, at least one square
  $\sigma(a,b,i,j)$ must be in the subcomplex $\K'$.

  Combining these results we see that if the subcomplex $\K'$ is a 2-sphere, then $\K'$
  contains exactly one square $\sigma(a_{i,j},b_{i,j},i,j)$ for each pair $i,j\in[n]$.
  Since for each values $i,j,a,b$ we have
  that $\sigma(a_{i,j},b_{i,j},i,j)\in \K$ if and only if $(a_{i,j},b_{i,j})\in S_{i,j}$,
  we obtain that $(a_{i,j},b_{i,j}) \in S_{i,j}$ for each $i,j\in [k]$.
  As $\tsquare(\sigma(a,b,i,j))$
  is identified with $\bsquare(\sigma(a',b,i-1,j))$, by induction we see that, for
  each $j\in[k]$, we have $b_{1,j}=b_{2,j}=\cdots =b_{k,j}$.
  A similar argument shows that for each $i\in[k]$ we have $a_{i,1}=a_{i,2}=\cdots=a_{i,k}$. 
  We deduce that $(a_{i,j},b_{i,j})$, for each pair $i,j\in[k]$,
  is a solution for $(n,k,\S)$.
\end{proof}

\Cref{lemma:gridtiling-to-2sphere} and~\cref{lemma:2sphere-to-gridtiling} establish correctness of our reduction from \gridtiling to \twosphere. Clearly, the reduction can be performed in polynomial time, and the parameter value $k'$ of a created instance is $\Oh(k^2)$. Thus,
\cref{theorem:main:lowerbound} now follows directly from \cref{theorem:marx:gridtilinglowerbound}.

%%%%%%%%%%%%%%%%%%%%%%%%%%%%%%%%%%%%%%%%%%%%%%%%%%%%%%%%%%%%%%%%%%%%%%%%%%%%%%%%%%%%%%%%%%%%%%%%%%%%%%%%%%%%%%%%%%%%%%
\section{A tight algorithm for 2-dim-sphere}\label{section:main:algorithm}

For each simplicial complex $\K$,
let $\Sd(\K)$ be its barycentric subdivision.
Its construction for the 2-dimensional case is as follows (see also \cref{fig:Sd}). 
Each vertex, edge and triangle of $\K$ is a vertex of $\Sd(\K)$.
To emphasize the difference, for a simplex $\tau$ of $\K$ we use
$v_\tau$ for the corresponding vertex in $\Sd(\K)$.
There is an edge $v_\tau v_{\tau'}$ in $\Sd(\K)$ between 
any two simplices $\tau$ and $\tau'$ of $\K$ 
precisely when one is contained in the other.
There is a triangle $v_{\tau_1}v_{\tau_2}v_{\tau_3}$ in $\Sd(\K)$
whenever there is a chain of inclusions $\tau_1\subsetneq \tau_2\subsetneq \tau_3$.
It is well-known, and not difficult to see, that $\Sd(\K)$ and $\K$ are homeomorphic.
See for example~\cite[Chapter 1]{Matousek07} or~\cite[Chapter 2]{Munkres93a}.
Let $\Sd_1(\K)$ be the 1-skeleton of $\Sd(\K)$, which is a graph.

\begin{figure}
\centering
	\includegraphics[page=5,width=\textwidth]{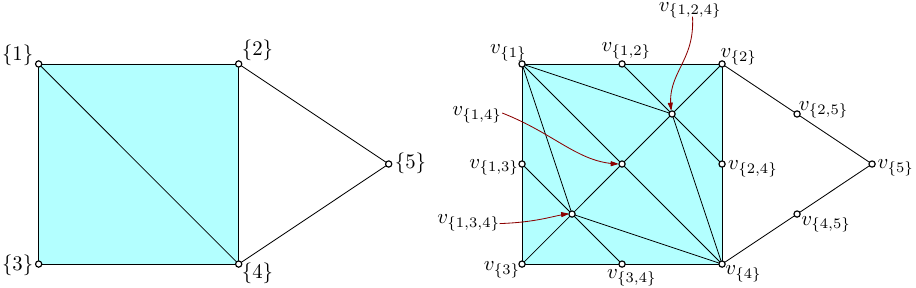}\label{fig:Sd}
	\caption{A simplicial complex $\K$ (left side) and its barycentric subdivision $\Sd(\K)$ 
		(right side).}
\end{figure}

An \DEF{isomorphism} between two simplicial complexes $\K_1$ and $\K_2$
is a bijective map 
$f\colon V(\K_1) \rightarrow V(\K_2)$ with the property that,
for all $\{v_1,\dots,v_k\}\subseteq V(\K_1)$,
the simplex $\{v_1,v_2,\dots ,v_k\}$ is in $\K_1$ precisely when 
$\{ f(v_1),f(v_2),\dots ,f(v_k)\}$ is a simplex of $\K_2$.
Two simplicial complexes are \DEF{isomorphic} if and only if
there exists some isomorphism between them.
When two simplicial complexes are isomorphic, 
they are also homeomorphic. (We can make geometric realizations
for both simplicial complexes with the same carrier.)
Note that isomorphism of simplicial complexes of dimension 1
matches the definition of isomorphism of graphs.

Testing isomorphism of simplicial complexes can be 
reduced to testing isomorphism of \emph{colored} graphs, as follows.
Let $G$ and $H$ be graphs and assume that we have 
colorings 
$c_G\colon V(G)\rightarrow \NN$ and 
$c_H\colon V(H)\rightarrow \NN$. 
(The term coloring here refers to a labeling; there is no relation
to the standard graph colorings.)
A \DEF{color-preserving isomorphism} between $(G,c_G)$ and $(H,c_H)$ is
a graph isomorphism $f\colon V(G) \rightarrow V(H)$ such that,
for each vertex $v\in V(G)$, it holds $c_H(f(v))=c_G(v)$.
Thus, the isomorphism preserves the color of each vertex.
We say that $(G,c_G)$ and $(H,c_H)$ are \DEF{color-preserving
isomorphic} if there is some color-preserving isomorphism between them. 
We will use the dimension $\dim$ of the simplex
as the coloring for the graph $\Sd_1(\cdot)$. Thus
$\dim(v_\tau)=|\tau|-1$ for each simplex $\tau$ of the simplicial complex.

\begin{lemma}
\label{lemma:isomorphism}
	Two simplicial complexes $\K_1$ and $\K_2$ are isomorphic
	if and only if $(\Sd_1(\K_1),\dim)$ and $(\Sd_1(\K_2),\dim)$ 
	are color-preserving isomorphic.
\end{lemma}

\begin{proof}
	Assume that $\K_1$ and $\K_2$ are isomorphic and let 
	$f\colon V(\K_1) \rightarrow V(\K_2)$ be an isomorphism.
	Then $f'\colon V(\Sd_1(\K_1)) \rightarrow V(\Sd_1(\K_2))$, defined
	by $f'(v_\tau)=v_{f(\tau)}$ for all $\tau \in \K_1$,
	is a color-preserving isomorphism
	between $(\Sd_1(\K_1),\dim)$ and $(\Sd_1(\K_2),\dim)$.
	
	Assume now that $f\colon V(\Sd_1(\K_1)) \rightarrow V(\Sd_1(\K_2))$
	is a color-preserving isomorphism
	between $(\Sd_1(\K_1),\dim)$ and $(\Sd_1(\K_2),\dim)$.
	The restriction of $f$ to the vertices of dimension $0$ of $\Sd_1(\K_1)$
	defines a map $f_0\colon V(\K_1) \rightarrow V(\K_2)$, where $f_0(a)$ is 
	the vertex with $f(v_a)=v_{f_0(a)}$. The map $f_0$ is a bijection
	because the restriction of $f$ to the vertices of dimension $0$ is a bijection.
	Next we argue that $f_0$ is an isomorphism.
	
	Consider a subset $\tau=\{ a_1,\dots,a_k\}$ of vertices in $\K_1$.
	If $\tau$ is a simplex of $\K_1$, then $f(v_\tau)=v_{\tau'}$ for
	some $\tau'\in \K_2$ . 
	Note that $|\tau'|=|\tau|=k$ because $f$ is a color-preserving isomorphism.
	For each $a_i\in \tau$, the edge $v_{a_i}v_\tau$ is in $\Sd_1(\K_1)$,
	and therefore $f(v_{a_i})f(v_\tau)=f(v_{a_i})v_{\tau'}$ is an edge of $\Sd_1(\K_2)$.
	It follows that $f(v_{a_1}),\dots, f(v_{a_k})$ are adjacent to $v_{\tau'}$.	
	Since by construction of $\Sd_1(\K_2)$ the vertex $v_{\tau'}$ 
	is adjacent to precisely $k$ vertices of $\Sd_1(\K_2)$ with dimension $0$, 
	those $k$ vertices are $f(v_{a_1}),\dots, f(v_{a_k})$.
	Since $f(v_{a_1})=v_{f_0(a_1)},\dots, f(v_{a_k})=v_{f_0(a_k)}$ because of
	the definition of $f_0$, 
	then $f_0(a_1),\dots, f_0(a_k)$ is a simplex of $\K_2$.	

	From the previous paragraph we conclude that, 
	if $\tau$ is a simplex of $\K_1$, then $f_0(\tau)$ is a simplex of $\K_2$.
	A symmetric argument shows that, if $f_0(\tau)$ is a simplex of $\K_2$,
	then $\tau$ is a simplex of $\K_1$. It follows that $f_0$ is an isomorphism between
	$\K_1$ and $\K_2$.
\end{proof}

\begin{lemma}
\label{lemma:subisomorphism}
	Let $\K_1$ and $\K_2$ be simplicial complexes.
	The simplicial complex $\K_1$ has a subcomplex isomorphic to $\K_2$
	if and only if $\Sd_1(\K_1)$ has a subgraph $G$ such that
	$(G,\dim)$ and $(\Sd_1(\K_2),\dim)$ are color-preserving isomorphic.
\end{lemma}
\begin{proof}
	Note that for each subcomplex $\K'_1$ of $\K_1$ we have
	that $\Sd_1(\K'_1)$ is exactly the subgraph of $\Sd_1(\K_1)$ induced
	be the vertices $v_{\tau}$, for $\tau\in \K'_1\setminus \{ \emptyset \}$.
	Therefore, if $\K_1$ has a subcomplex $\K'_1$ isomorphic to $\K_2$,
	then the graph $G=\Sd_1(\K'_1)$ is a subgraph of $\Sd_1(\K_1)$ and,
	by \cref{lemma:isomorphism}, $(G,\dim)$ and $(\Sd_1(\K_2),\dim)$ 
	are color-preserving isomorphic.
	
	Assume, for the other direction, 
	that $(G,\dim)$ and $(\Sd_1(\K_2),\dim)$ are color-preserving isomorphic
	for some subgraph $G$ of $\Sd(\K_1)$. 
	Let $f$ be such a color-preserving isomorphism.
	First we show that $G$ is $\Sd_1(\K'_1)$ for some subcomplex $\K'_1$ of $\K_1$.
	Indeed, consider any vertex $v_\tau$ of $G$ such that for no superset $\tilde\tau$ of $\tau$ 
	we have $v_{\tilde\tau}$ in $G$.
	The vertex $f(v_\tau)$ is a vertex $v_{\tau_2}$ for $\tau_2\in \K_2$ and
	moreover $|\tau|=|\tau_2|$ as $f$ preserves color and therefore dimension.
	Each subset $\tau'_2$ of $\tau_2$ has some vertex $v_{\tau'_2}$
	in $\Sd_1(\K_2)$. For each such $\tau'_2\subset \tau_2$ we have some
	distinct vertex $v_{\tau'_1}$ in $G$ such that $f(v_{\tau'_1})=v_{\tau'_2}$
	and $v_{\tau'_1}$ must be adjacent to $v_\tau$.
	Since $\tau$ and $\tau_2$ have the same cardinality,
	they have the same number of subsets, and thus $\tau'_1$ iterates over 
	all subsets of $\tau$, when $\tau'_2$ iterates over the subsets of $\tau_2$. 
	This means that $v_{\tau'}$ is in $G$ for all subsets $\tau'\subset\tau$.
	Therefore, if we take $\K'_1=\{\tau\in \K_1\mid v_\tau \in V(G) \}\cup \{ \emptyset \}$,
	then $\K'_1$ is a simplicial complex and $G=\Sd_1(\K'_1)$. 
	From \cref{lemma:isomorphism} it follows that $\K'_1$ and $\Sd_1(\K_2)$ are isomorphic.
\end{proof}

\begin{lemma}
\label{lemma:color-coding}
	Let $\K$ be a simplicial complex with $n$ simplices
	and let $\K'$ be a simplicial complex with $k$ simplices.
	Let $t$ be the treewidth of $\Sd_1(\K')$.
	In time $2^{\Oh(k)} n^{\Oh(t)}$ we can decide whether $\K$ contains
	a subcomplex isomorphic to $\K'$.
\end{lemma}
\begin{proof}
	Alon, Yuster and Zwick~\cite[Theorem 6.3]{AlonYZ95} show how to find
	in a graph $G$ a subgraph isomorphic to a given graph $H$ in
	time $2^{\Oh(|V(H)|)} |V(G)|^{\Oh(t_H)}$, where $t_H$ is the treewidth of $H$.
	The technique is color-coding, which has become a standard tool for
	developing fixed-parameter algorithms;
	see for example~\cite[Section 5.2]{CyganFKLMPPS15}. 
	In this technique, one tries several different
	colorings of the vertices of $G$ with $|V(H)|$ colorings, and then
	uses dynamic programming to search for a copy of $H$ in $G$ where
	all the colors of the vertices are distinct. Thus, if the vertices of $H$
	and $G$ are already classified into some classes, then this can only
	help the algorithm. The class of a vertex can be considered as
	part of the coloring. This means that the algorithm can be trivially
	adapted to the problem of subgraph color-preserving isomorphism:
	given two pairs $(G,c_G)$ and $(H,c_H)$, where $c_G$ and $c_H$ are colorings
	of the vertices, is there a subgraph $G'$ of $G$ such that
	$(G',c_{G'})$ and $(H,c_H)$ are color-preserving isomorphic, where
	$c_{G'}$ is the restriction of $c_G$ to $G'$.
	
	Because of \cref{lemma:subisomorphism}, 
	deciding whether $\K$ contains a subcomplex isomorphic to $\K'$
	is equivalent to deciding whether
	$\Sd_1(\K)$ contains a subgraph $G$ such that
	$(G,\dim)$ and $(\Sd_1(\K'),\dim)$ are color-preserving isomorphic.
	Apply the color-coding algorithm of Alon et al.~as, discussed before,
	we spend $2^{\Oh(|V(\Sd_1(\K'))|)} |V(\Sd_1(\K))|^{\Oh(t)} = 2^{\Oh(k)} n^{\Oh(t)}$ time.
\end{proof}

\begin{proof}[Proof of \cref{theorem:main:upperbound}]
	There are $2^{\Oh(k)}$ different (unlabeled) triangulations of the 2-sphere
	with at most $k$ triangles; see for example~\cite{Tutte,BonichonGHPS06}, using that $k$ triangles entail having at most $\Oh(k)$ vertices. 
	For each such triangulation, let $\K_i$ be the
	corresponding simplicial complex. Then $\Sd_1(\K_i)$ is a planar graph
	with $\Oh(k)$ vertices and thus has treewidth $\Oh(\sqrt{k})$.
	Using \cref{lemma:color-coding} we can decide in $2^{\Oh(k)}n^{\Oh(\sqrt{k})}$ time
	whether $\K$ has a subcomplex isomorphic to $\K_i$.
	Iterating over all the $2^{\Oh(k)}$ triangulations we spend
	in total 
	$2^{\Oh(k)}\cdot 2^{\Oh(k)} n^{\Oh(\sqrt{k})}$ time and the result follows.
\end{proof}

%%%%%%%%%%%%%%%%%%%%%%%%%%%%%%%%%%%%%%%%%%%%%%%%%%%%%%%%%%%%%%%%%%%%%%%%%%%%%%%%%%%%%%%%%%%%%%%%%%%%%%%%%%%%%%%%%%%%%%
\section{Kernelization and compression for Deletion-to-2-dim-sphere}\label{section:preprocessing}

In this section, we prove that \deletiontotwosphere admits a polynomial kernelization that returns instances with $\Oh(k^2)$ triangles and has bit-size $\Oh(k^2\log k)$, respectively a polynomial compression to a weighted version with bit-size $\Oh(k\log k)$. We first give a few simple reduction rules and then show how to reduce (and possibly encode) the resulting instances. The rules are to be applied in order, i.e., preference is given to earlier rules. Recall that input instances $(\K,k)$ consist of a $2$-dimensional simplicial complex \K and an integer $k$, and ask whether deletion of at most $k$ triangles from \K yields a subcomplex that is homeomorphic to the $2$-dimensional sphere $\SS_2$.

In what follows, we will delete subcomplexes from an instance of our problem and at the same time reduce the value of $k$. If at any point in time $k$ becomes negative we know that our original instance was a no-instance, so we will assume that $k$ is always non-negative.
Additionally we point out that whenever deleting a subcomplex from our simplicial complex, any vertices or edges which would no longer be contained in any triangle are also deleted.

\begin{rrule}\label{rule:degreeone}
If any triangle $T\in\K$ has an edge that is not an edge of any other triangle in $\K$ then delete $T$ from $\K$ and reduce $k$ by one.
\end{rrule}

Clearly, such a triangle $T$ cannot be contained in a subcomplex $\K'\subseteq\K$ that is homeomorphic to the $2$-sphere, and hence it must be among the $k$ deleted triangles in any solution (if one exists).
Note that when \cref{rule:degreeone} does not apply, each edge in $\K$ is shared by at least two triangles of \K. On the other hand, in the desired subcomplex that is homeomorphic with the $2$-sphere each edge is shared by \emph{exactly} two triangles. Denote by $\T\subseteq\K$ the set of triangles that share at least one of their edges with more than one other triangle. There is a simple upper bound for the size of $\T$ if $(\K,k)$ is a yes-instance.

\begin{proposition}
If $(\K,k)$ is a yes-instance of \deletiontotwosphere then $|\T|\leq 7k$.
\end{proposition}

\begin{proof}
  Let $\D$ be a given solution with at most $k$ triangles; this means that $\K\setminus \D$
  is a 2-sphere.
  Each triangle in $\T\setminus\D$ must share at least one edge with a triangle in $\D$.
  Additionally, each triangle in $\D$ can share an edge with at most six triangles in $\T\setminus\D$, as each of the three edges of a triangle in $\D$ is shared between at most two triangles of $\T\setminus\D$.
  Thus, $|\T\setminus\D| \leq 6 \cdot |\D|$, giving $|\T| \leq 7k$.
\end{proof}

\begin{rrule}\label{rule:rejectlargec}
Reject the instance if $|\T|>7k$.
\end{rrule}

Observe now that in $\K_\T:=\K\setminus\T$ all edges are shared by at most two triangles, and that edges shared with triangles in $\T$ are only part of \emph{one} triangle in $\K_\T$. 
Let us say that a simplicial complex $\K$ is \DEF{edge-connected} if between
each two points of $\K$ there exists a path whose interior is disjoint from the vertices of $\K$. Thus,
an edge-connected simplicial complex $\K$ is obtained by gluing triangles along edges such that
one gets a connected simplicial complex.
We can also identify vertices, but that operation does not affect whether it is thick-connected or not.
It is easy to see that one can define edge-connected components as maximal edge-connected subcomplexes.
Some vertices may belong to several edge-connected components. See \cref{fig:components} for
an example.
In the following, \emph{whenever we refer to a component, it means an edge-connected component}.
Accordingly, triangles in $\K_\T$ form components that can be homeomorphic to, e.g., the $2$-sphere or to a punctured disk. Say that the \emph{boundary} of a component is the set of edges $L$ that are contained in exactly one triangle of the component; these are exactly the edges that participate also in triangles of $\T$. We distinguish components according to whether or not they have a boundary.

\begin{figure}
\centering
	\includegraphics[page=6,scale=.7]{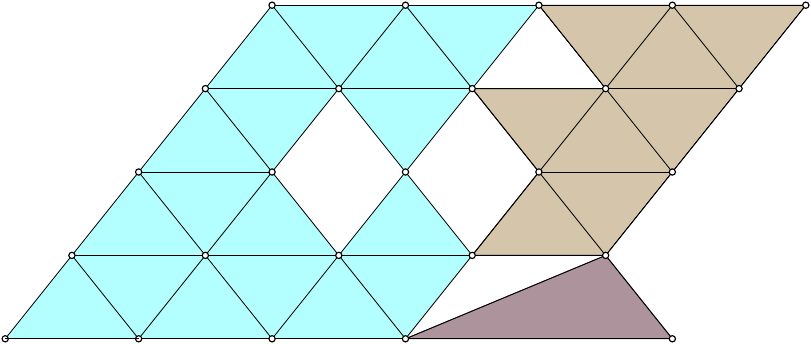}
	\caption{A simplicial complex with $3$ (edge-connected) components, each marked with different colors.}
	\label{fig:components}
\end{figure}

For any component without boundary the procedure is simple: It cannot have any edge $L$ in common with a triangle of $\T$ since then three or more triangles of $\K$ would share $L$ and all incident triangles would be in $\T$. Accordingly, such a component can only be part of the desired $2$-sphere if it itself is homeomorphic to the $2$-sphere since it is not connected with other triangles in $\K$. The only other option is to delete the entire component since deleting it partially would always leave a boundary.

\begin{rrule}\label{rule:componentwithoutboundary}
Let $\C$ be a component of $\K_\T$ that has no boundary. If $\C$ is homeomorphic to the $2$-sphere $\SS_2$ and $|\K\setminus\C|\leq k$ then answer yes (and return $\K\setminus\C$ as a solution). Else, if $\C$ is not homeomorphic to the $2$-sphere or if $|\K\setminus\C|>k$, then delete all triangles of $\C$ from $\K$ and reduce $k$ by $|\C|$.
\end{rrule}

Using \cref{rule:degreeone} through~\cref{rule:componentwithoutboundary} we either solve the instance or we arrive at the situation where $|\T|\leq 7k$ and all components of $\K_\T=\K\setminus\T$ have boundaries. Observe that, among these, we can safely delete each component $\C$ that is not homeomorphic to a (punctured) disk: Such a subcomplex $\C$ cannot be extended to a subcomplex of $\K$ that is homeomorphic to the $2$-sphere because the requirement of having two triangles incident with each edge implies using all triangles of $\C$. For example, when $\C$ is a punctured torus, we cannot extend it to a sphere using the whole $\C$.

\begin{rrule}\label{rule:componentwithboundarybutnotpunctureddisk}
If $\C$ is a component of $\K_\T$ that has a boundary but is not homeomorphic to a (punctured) disk then delete all triangles of $\C$ from $\K$ and reduce $k$~by~$|\C|$.
\end{rrule}

It remains to consider the case where $|\T|\leq 7k$ and all components of $\K_\T$ (have boundaries and) are homeomorphic to (punctured) disks. As a first step, let us observe an upper bound on the total length of all component boundaries (in terms of number of edges) for yes-instances.

\begin{proposition}\label{proposition:boundarybound}
If $(\K,k)$ is a yes-instance of \deletiontotwosphere then the total length of all boundaries of components of $\K_\T$ is at most $21k$.
\end{proposition}

\begin{proof}
By \cref{rule:degreeone} each boundary edge of a component of $\K_\T=\K\setminus \T$ is incident with at least two triangles of $\K$, and hence with at least one triangle of $\T$. The upper bound of $3\cdot|\T|\leq 21k$ follows.
\end{proof}

Note that from the upper bound of $21k$ for the total boundary length we immediately get an upper bound of $7k$ for the number of components of $\K_\T$ since each component with a boundary must have at least three boundary edges. To get an upper bound on the number of triangles it now suffices to replace large components by ``equivalent'' ones without changing the status of the instance using \cref{le:componentreplacement}. This has two vital aspects: (1) Replaced components must have the same boundary and topology. (2) We must avoid creating false positives, as smaller components can be deleted at a lower cost. In \cref{le:componentreplacement} we show how components with boundary length $\ell$ can be replaced by equivalent ones with $\Oh(\ell)$ triangles, addressing (1), and later give two options for addressing~(2). 

\begin{lemma}\label{le:componentreplacement}
  Given a simplicial complex $\K$ of a punctured sphere where $\K$ contains $\ell$ boundary edges, there exists a simplicial complex $\K'$ such that the following hold:
  \begin{enumerate}
  \item $\K'$ contains $\Oh(\ell)$ triangles,
  \item $\K$ is homeomorphic to $\K'$,
  \item $\K$ and $\K'$ have exactly the same boundary, and
  \item if $a$ and $b$ are edges of $\partial(\K)$ such that there exists a triangle $t$ of $\K$ with $a,b\in t$, then there exists a triangle $t' \in \K'$ such that $a,b \in t'$.\label{adjacentedges}
  \end{enumerate}
\end{lemma}
\begin{proof}
  Note that if $a$ and $b$ are edges of $\K$ as in condition~\ref{adjacentedges}, then in $\K$ edges $a,b$ must share a common vertex also on the boundary of $\K$.
  If $\ell = 3$, then $\K$ is simply a triangle and $\K'=\K$ suffices.
  If $\ell > 3$, we will construct $\K'$. We will denote by $v$ a central vertex of $\K'$. The boundary of $\K'$ is the boundary of $\K$.
  For any $a,b$ as in condition~\ref{adjacentedges}, let $a,b = ((v_i,v_j),(v_j,v_k))$ and insert the triangles $(v, v_i, v_j)$ and $(v_i, v_j, v_k)$ into $\K'$.
  For any remaining edges $(v_i, v_j)$ in $\partial(\K)$, add the triangle $(v, v_i, v_j)$ to $\K'$.
  Note that no edge $(v, v_i)$ will be on the boundary of $\K'$, as each $v_i$ must be the intersection of exactly two edges of $\partial(\K)$. We have a potential problem that the we may introduce several ``parallel'' edges $vv_i$, which is not allowed in a simplicial complex. We can fix this by subdividing each edge $vv_i$ of which we try to make multiple copies, and retriangulating the two faces incident to those edges. See \cref{fig:componentreplacement} for an example. Standard arguments often used for planar graphs show that this procedure constructs the desired punctured sphere (or disc) using $\Oh(\ell)$ triangles.
\end{proof}

\begin{figure}
\centering
	\includegraphics[page=7]{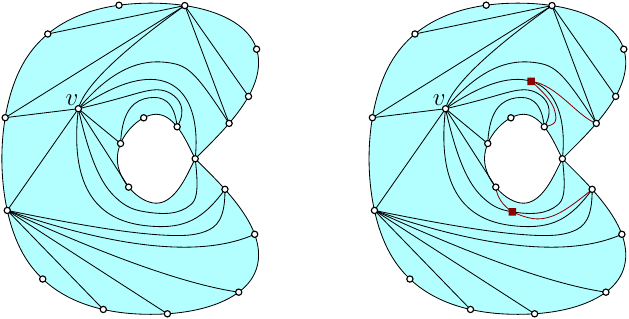}
	\caption{Proof of \cref{le:componentreplacement}. On the left side we have ``parallel'' edges and we handle it subdividing those parallel edges and triangulating the incident faces.}
	\label{fig:componentreplacement}
\end{figure}

To avoid false positives we have two options. First, we can store for each component its initial number of triangles, i.e., the cost for deleting it entirely, noting that costs larger than $k$ can be replaced by $k+1$. (Recall that partially deleting a component is infeasible.) The output would then be an instance of a weighted version of the problem, and we could encode it using $\Oh(k\log k)$ bits, where the log-factor is needed to encode costs in binary and to represent a list of the triangles including vertex names. (We could also assign a larger cost to one triangle per component such that the total is equal to the original value.)

Second, we could apply the replacement only to components with more than $k$ triangles, and afterwards increase their size to $k+\Oh(1)$ by adding additional triangles. Since budget of $k$ does not allow the deletion of large components, this yields an equivalent instance. The total number of triangles per component is then $\Oh(k)$, and $\Oh(k^2)$ for the entire instance; this can be encoded in $\Oh(k^2\log k)$ bits. This completes the proof of \cref{theorem:other:compression}.

The compression result can be lifted to a smaller parameter, namely the number $t$ of \emph{conflict triangles}, i.e., triangles incident with at least one edge that is shared by at least three triangles (\cref{corollary:compression:smallerparameter} below). To see that this is a stronger parameter, recall that nontrivial instances with budget $k$ have $\Oh(k)$ conflict triangles, and observe that having few conflict triangles does not bound the size $k$ of the desired 2-sphere. 

Observe that for parameter $t$ there is a simple $\Oh^*(2^t)$ time algorithm: First, guess by complete enumeration which of the $t$ conflict triangles are to be deleted. Reject a guess if an edge with at least three incident triangles remains. Iteratively delete triangles that uniquely contain any edge. Thus, we arrive at possibly several components where all edges are shared by exactly two triangles each. Determine which components are homeomorphic to the $2$-sphere and reject the guess if there is no such component. Accept if deleting all triangles outside the largest component costs only $k$ deletions in total (including prior deletions); else reject the guess.

\begin{corollary}\label{corollary:compression:smallerparameter}
The \deletiontotwosphere problem admits a polynomial compression to weighted instances with $\Oh(t)$ triangles and bit-size $\Oh(t^2)$ where $t$ is the number of conflict triangles in the input.
\end{corollary}

\begin{proof}
Let $(\K,k)$ be an instance of \deletiontotwosphere with $t$ conflict triangles. All four reduction rules can be safely applied: They preserve the correct yes- or no-answer and they do not increase (but may decrease) the number of conflict triangles. Through exhaustive application of the rules we obtain an equivalent instance $(\K',k')$ with a set $\T$ of $t'\leq t$ conflict triangles such that all components of $\K'\setminus\T$ are homeomorphic to punctured discs. Following previous arguments there are at most $|\T|\leq t$ such components and their total boundary length is at most $3|\T|\leq 3t$. It remains to encode the resulting instance into a number of bits that is polynomial in $t$.

We can use a standard trick for this: If the total number of triangles, say $n$, is at least $2^t$, we can solve the instance in polynomial time using the aforementioned FPT-algorithm. Else, we have $n< 2^t$ and, hence, numbers of value up to $n$ cost only $t$ bits to encode in binary. For each component of $\K'\setminus\T$ we apply~\cref{le:componentreplacement} to obtain a total of $\Oh(t)$ triangles that represent components with the same topology and same boundaries. Additionally, assign weights to the triangles such that the total weight of each component is equal to its number of triangles before the replacement. These weights take at most $t$ bits each, for a total size of $\Oh(t^2)$ bits.
\end{proof}

\section{Conclusion}\label{section:conclusion}

Our hardness results can be extended easily to cases of finding some other surfaces, such as a torus. Indeed, we can replace in the construction the back sheet with any other shape that has the target topology.
Similarly, the positive results can also be extended to the search for small surfaces, again like the torus.

It is clear that the simplicial complex we use to show hardness cannot be embedded in 3-dimensional space. It is unclear how hard the problem \twosphere is when restricted to simplicial complexes that are embedded in $\RR^3$. Note that it is not meaningful to parameterize the problem by the dimension of some ambient space because any $2$-dimensional simplicial complex can be embedded in $\RR^5$ using the moment curve~(\cite[Section 1.6]{Matousek07}).

A simplicial complex can be generalized to something called, unsurprisingly, a generalized triangulation (or sometimes just referred to as a triangulation).
In this setting, we are allowed to identify facets of a common simplex. That is, in a 2-dimensional generalized triangulation we may identify together two distinct edges of the same triangulation.
This relaxation can make it harder to even detect a manifold, as there are more cases to consider.
Our work here, and related problems, are all still of interest in this setting.

Lastly, the problems discussed in this paper generalize, where possible, in the obvious manner to higher dimensions.
In particular, fast detection of $3$-sphere subcomplexes (or sub-triangulations) that do not bound a ball are of particular interest for the recognition of the prime decomposition of $3$-manifolds.

\subparagraph*{Acknowledgments.}
This work was initiated during the Fixed-Parameter Computational Geometry Workshop 
at the Lorentz Center, 2016. We are grateful to the other participants of the workshop
and the Lorentz Center for their support.

% \bibliographystyle{abbrv}
% \bibliography{bibliography_arxiv}

\end{document}